\documentclass[english, letterpaper, 10 pt, conference]{ieeeconf}
\usepackage[T1]{fontenc}
\usepackage[latin9]{inputenc}
\usepackage{amsmath}
\usepackage{cite}
\usepackage[caption=false]{subfig}
\usepackage{dsfont}
\usepackage{bm}
\usepackage{url}
\usepackage{balance}

\IEEEoverridecommandlockouts
\makeatletter

\newtheorem{defn}{Definition}

\usepackage{tikz}
\usepackage{tikzscale}
\usepackage{amssymb}

\usepackage{enumerate}   
\usepackage{graphicx} 
\usepackage[binary-units=true]{siunitx}
\usepackage{pgfplots}
\pgfplotsset{compat = 1.12}
\usepackage{cleveref}
\usepackage{algorithm}
\usepackage{algorithmicx}
\usepackage{algpseudocode} 
\usepackage{booktabs}
\usepackage{xspace}
\usetikzlibrary{external}
\usepgfplotslibrary{groupplots}
\usetikzlibrary{matrix}

\newcounter{tagnumb}
\newcommand{\bT}{\mathbf T}
\newcommand{\bN}{\mathbf N}
\newcommand{\PP}{\mathcal P}
\newcommand{\QQ}{\mathcal Q}
\newcommand{\RR}{\mathcal R}
\newcommand*{\ev}{{ego}-vehicle\@\xspace}
\newcommand*{\ie}{\textit{i.e.}\@\xspace}
\newcommand*{\eg}{\textit{e.g.}\@\xspace}
\newcommand{\twodots}{\mathinner {\ldotp \ldotp}}
\newcommand{\threedots}{\dots}

\DeclareMathOperator{\adj}{adj}
\DeclareMathOperator{\Adj}{Adj}

\newtheorem{theorem}{Theorem}

\newtheorem{corollary}{Corollary}

\crefname{defn}{def.}{defs.}
\Crefname{defn}{Definition}{Definitions}
\crefname{algocf}{alg.}{algs.}
\Crefname{algocf}{Algorithm}{Algorithms}

\newcommand{\includegtikz}[2][]{ 
	\tikzsetnextfilename{#2}
	\includegraphics[#1]{figures/tikz/#2} }

\algrenewcommand\algorithmicindent{1.0em}%

\makeatother

\usepackage{babel}
\usepackage{diagbox}

\title{\LARGE \bf Partitioning of the Free Space-Time for On-Road Navigation of Autonomous Ground Vehicles}

\author{Florent Altch\'e$^{2,1}$ and Arnaud de La Fortelle$^{1}$
\thanks{$^{1}$ MINES ParisTech, PSL Research University, Centre for robotics, 60 Bd St Michel 75006 Paris, France {\tt\small [florent.altche, arnaud.de\_la\_fortelle] @mines-paristech.fr}}
\thanks{$^{2}$ \'Ecole des Ponts ParisTech, Cit\'e Descartes, 6-8 Av Blaise Pascal, 77455 Champs-sur-Marne, France}%
}

\begin{document}

\maketitle
\thispagestyle{empty}
\pagestyle{empty}

\begin{abstract}In this article, we consider the problem of trajectory planning and control for on-road driving of an autonomous ground vehicle (AGV) in presence of static or moving obstacles. We propose a systematic approach to partition the collision-free portion of the space-time into convex sub-regions that can be interpreted in terms of relative positions with respect to a set of fixed or mobile obstacles. We show that this partitioning allows decomposing the NP-hard problem of computing an optimal collision-free trajectory, as a path-finding problem in a well-designed graph followed by a simple (polynomial time) optimization phase for any quadratic convex cost function. Moreover, robustness criteria such as margin of error while executing the trajectory can easily be taken into account at the graph-exploration phase, thus reducing the number of paths to explore.
\end{abstract}

\section{Introduction}
In order to drive on public roads, autonomous ground vehicles (AGVs) will be required to navigate efficiently inside a potentially dense flow of other vehicles with uncertain behaviors. For this reason, planning safe, efficient and dynamically feasible trajectories that can be safely followed by a low-level controller is a particularly important problem.

One of the difficulties of ``optimal'' trajectory planning for AGVs is that the presence of obstacles renders the search space non-convex, and multiple possible maneuver variants (for which there is at least one locally optimum trajectory~\cite{Bender2015}) exist, such as illustrated in \Cref{fig:example-sit}. At control level, tracking the computed trajectory may involve highly nonlinear vehicle dynamics  when nearing the handling limits (see, \eg,~\cite{Kong2015} for a review); in less demanding (low-slip) scenarios, simpler dynamic models can be used.

Because they allow simultaneous trajectory generation with obstacle avoidance and control computation, model predictive control (MPC) approaches have been very popular for AGVs (see, \eg,~\cite{Falcone2007b,Abbas2014}). However, real-time constraints usually force authors considering very precise dynamic models to choose a short (sub-second) prediction horizon, which may in turn cause the MPC problem to become infeasible, for instance when a new obstacle is detected with not enough time to stop. Even with simpler dynamic models, the non-convexity of the state-space renders continuous optimization techniques inefficient. For this reason, hierarchical frameworks~\cite{Cowlagi2012,qian2016hierarchical} have been proposed, in which a medium-term (up to a dozen seconds) planner generates a rough trajectory which is then refined by a short-term (sub-second to a few seconds) controller. Mixed-integer programming (MIP) methods are often used in medium-term trajectory planning to encode the discrete decisions arising from multiple maneuver choices~\cite{Schouwenaars2002,Anderson2011}, generalized as \textit{logical constraints} in~\cite{Qian2016}. However, MIP problems are known to be NP-hard~\cite{karp1972reducibility} and are therefore difficult to solve in real-time.

\begin{figure}
	\centering \includegtikz[width=0.8\columnwidth]{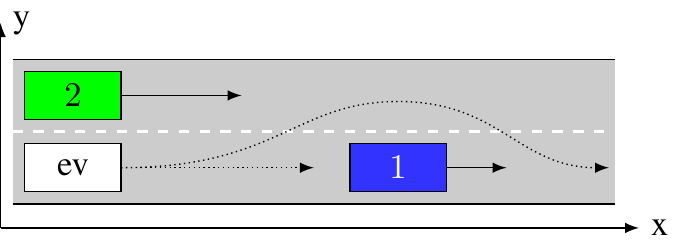}
	\caption{Example driving situation involving multiple maneuver choices for an AGV (denoted ev): overtake the slower (blue, denoted $1$) vehicle before the green vehicle ($2$) passes or wait behind the blue vehicle, possibly overtaking after the green vehicle has passed. Solid arrows represent the velocity of each vehicle, dotted arrows represent possible AGV trajectories. \label{fig:example-sit}}
\end{figure}

In this article, we propose a different approach for maneuver selection, inspired by the use of graph-based coordination of robots~\cite{Gregoire2014} and the decomposition of the collision-free space presented in~\cite{Park2015} for 2D path-planning. First, we introduce a systematic algorithm to partition the collision-free space-time into 3D regions with geometrical adjacency relations; the structure of on-road driving allows to assign a semantic interpretation to each partition subset. Using time discretization, we further divide these regions into convex polyhedrons, and design a \textit{transition graph} in which any path corresponds to a collision-free trajectory (that may however be dynamically infeasible). The main advantage of our approach is to reduce the entire combinatorial decision-making process (choosing from which side to avoid each obstacle) to the selection of a path in a graph. Once such a path has been selected, we show that computing a corresponding optimal trajectory (for a quadratic convex cost function) in an MPC fashion is widely simplified and can be performed in polynomial time. 

A second advantageous property of our decomposition approach is to simplify the use of risk metrics, which can be directly taken into account at the graph exploration phase; in~\cite{Constantin2014}, the authors used a similar partitioning technique to design a ``space margin'' metric for 2D path planning. In this article, we introduce a complementary \textit{time margin} metric, corresponding to a temporal tolerance to execute a particular maneuver, that can be easily computed from our graph representation. This measure is related to the notion of ``gap acceptance'', commonly used in stochastic decision-making (see, \eg,~\cite{Lefevre2012}). We believe that combining a temporal margin (notably accounting for uncertainty in predicting the future trajectory of moving obstacles) as well as a spatial margin (accounting for perception and control errors) is key for trajectory planning and tracking in real-world situations, for instance coupled with MPC or Linear Quadratic Gaussian motion planning and control~\cite{VandenBerg2011}.

Our transition graph approach generalizes state-machine-based techniques~\cite{Wang2015,Wang2016} which rely on a predefined set of maneuvers (such as \textit{track lane} or \textit{change lane}) that needs to be manually adapted to the driving situation. By contrast, our method can be applied in many scenarios (including highway and urban driving, for instance crossing an intersection) with the same formalism. Although spatio-temporal graphs have already been used for the control of AGVs~\cite{Ono2015,TianyuGu2015}, no existing approach provides the same desirable properties, and notably to easily account for margins in planning.


The rest of this article is structured as follows: in \Cref{sec:example}, we present intuitions of our main ideas using the example scenario of \Cref{fig:example-sit}. In \Cref{sec:math}, we formalize these intuitions mathematically, and we present applications of our results to planning and control for autonomous ground vehicles in \Cref{sec:optimization}. In \Cref{sec:simulation}, we present early simulation results and data on computation time; finally, \Cref{sec:conclusion} concludes the study.

\section{A guiding example}\label{sec:example}
The goal of this section is to give an intuition of our main mathematical results using the example scenario shown in \Cref{fig:example-sit}; the formal mathematical theory is developed in the next section. In our example, we consider an autonomous ground vehicle (called \ev in the remainder of this article) navigating on a road with two other vehicles (obstacles); vehicles are modeled as rectangles driving parallel to the side of the road. Intuitively, the \ev has three classes of maneuvers to choose from: either it can remain behind vehicle $1$, overtake it before vehicle $2$ passes, or overtake it after vehicle $2$ has passed; in~\cite{Bender2015}, these maneuver choices are linked to the notion of homotopy classes of trajectories. Assuming that the future trajectory of the obstacles is known in advance, it is possible to compute the \textit{obstacle set} $\chi_o$ of $(x,y,t)$ positions of the \ev for which a collision exists at time $t$; the complement of this set is the \textit{collision-free} region of the space-time (or free space-time), denoted by $\chi_f$. Any collision-free trajectory for the \ev corresponds to a path in $\chi_f$; \Cref{fig:free-tspace} provides an illustration of the free space-time in our example.

\begin{figure}
	\centering \subfloat{\includegtikz[width=.7\columnwidth]{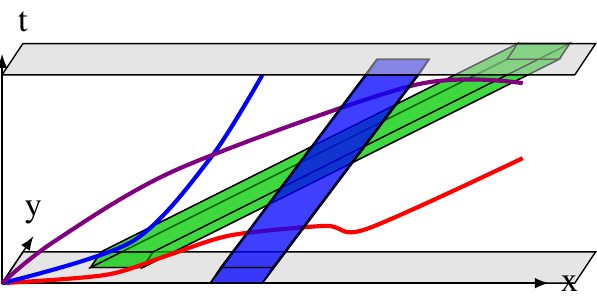}} \subfloat{\includegtikz[width=.25\columnwidth]{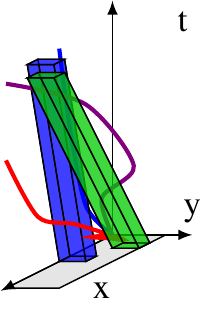}}
	\caption{Free space-time $\chi_f$ (in white) corresponding to the situation of \Cref{fig:example-sit}. Obstacles are pictured in the color of the corresponding vehicle. Light-gray planes represent the road extent in the $y$ direction. The thick curves represent possible collision-free trajectories for the \ev. \label{fig:free-tspace}}
\end{figure}

Due to the complex structure of the free space-time, notably its non-convexity, this abstraction is difficult to use directly to compute optimal collision-free trajectories. Inspired by the work in~\cite{Gregoire2014} and \cite{Park2015}, we propose a decomposition of $\chi_f$ in convex subregions with adjacency relations. First, we partition horizontal planes (corresponding to fixed time instants) using relative positions with respect to each obstacle as illustrated in \Cref{fig:partition}. Each subset of the partition corresponds to positions where the \ev is either located in front ($f$), to the left ($l$), behind ($b$) or to the right ($r$) of each obstacle. Using the additional information given by road boundaries, this partitioning technique yields four subsets denoted by $(lb)$, $(lf)$, $(br)$ and $(fr)$, indicating the relative position of the \ev from obstacle $1$ and $2$ in this order. We call these labels \textit{signature} of each subset. Additionally, for two such subsets $A$ and $B$ at a given time $t$, we can define an adjacency relation $\adj_t$ (related to that of~\cite{Park2015}), such that $\adj_t(A,B) = 1$ if the intersection of their closures is not empty, \ie $\bar A \cap \bar B \neq \emptyset$. 

\begin{figure}\centering
	\subfloat[$t = t_0$ : $\adj_{t_0}( lf, br ) = 1$, $\adj_{t_0}(lb, fr) = 0$ ]{\includegtikz[width=\columnwidth]{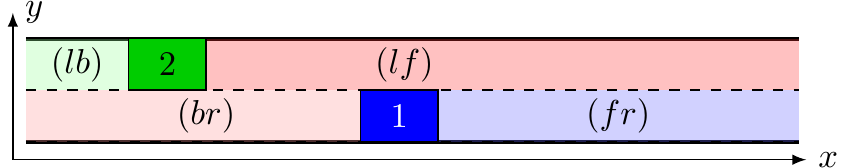}}
	
	\centering
	\subfloat[\label{fig:partition3b}$t = t_1$ : $\adj_{t_1}(br, lf) = \adj_{t_1}(lb, fr) = 0$]{\includegtikz[width=\columnwidth]{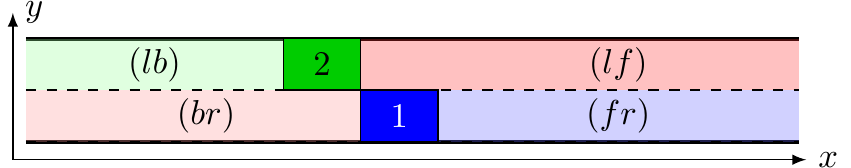}}
	
	\centering
	\subfloat[$t = t_2$ : $\adj_{t_2}(br, lf) = \adj_{t_2}(lb, fr) = 0$]{\includegtikz[width=\columnwidth]{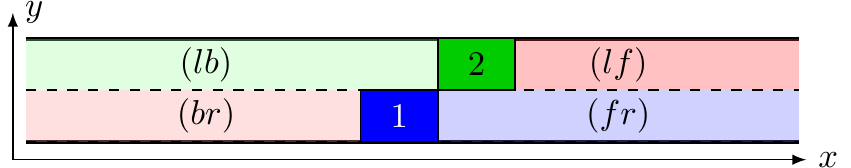}}
	
	\centering
	\subfloat[$t = t_3$ : $\adj_{t_3}( lf, br ) = 0$, $\adj_{t_3}(lb, fr) = 1$]{\includegtikz[width=\columnwidth]{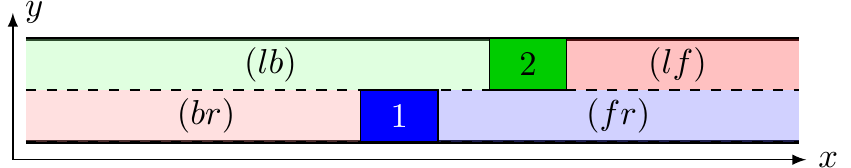}}
	\caption{Partitioning of the 2D space at different times in our example scenario, and adjacency relations $\adj$. In this example, $\adj_t(lb, br) = \adj_t(lf, fr) = 1$ and $\adj_t$ is symmetrical at all times. \label{fig:partition}}
\end{figure} 

This partitioning method can be generalized to the three-dimensional space-time by using unions of regions sharing the same signature, as shown in \Cref{fig:partition-3d}. The notion of adjacency described above can be extended, and we let $\Adj(A,B)$ be the set of times $t$ such that $\adj_t(A,B) = 1$. We call the set $\Adj(A,B)$ the \textit{validity set} of the transition from $A$ to $B$, corresponding to time periods for which a collision-free trajectory from $A$ to $B$ exists. The validity sets in this example are given in \Cref{tab:adjacency-set}, with initial time $t_0$.

\begin{figure}
	\centering \subfloat{\includegtikz[width=\columnwidth]{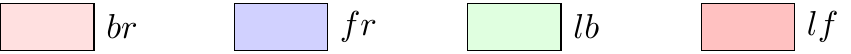}}
	
	\vspace*{-0.2cm}	
	\centering \subfloat{\includegtikz[width=.7\columnwidth]{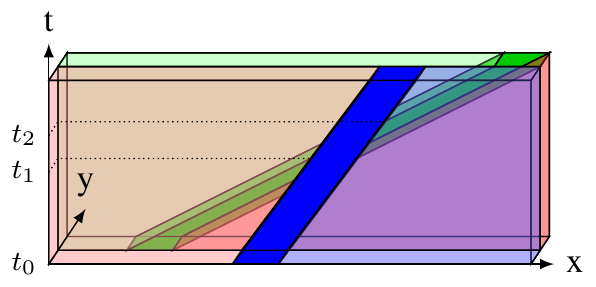}}
	\subfloat{\includegtikz[width=.25\columnwidth]{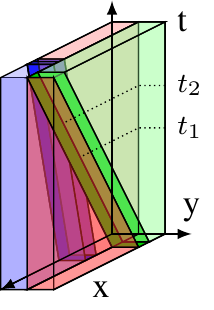}}
	\caption{Partitioning of the free space-time of \Cref{fig:free-tspace} into four cells. The legend gives the signature of each cell, with blue obstacle first. \label{fig:partition-3d}}
\end{figure}

\begin{table}
	\centering
	\caption{Validity sets $\Adj(A,B)$ \label{tab:adjacency-set}}
	\begin{tabular}{c c c c c}
		 & $br$ & $fr$ & $lb$ & $lf$ \\
		\cmidrule(lr){2-5} 
		$br$ & $[t_0, +\infty)$ & $\emptyset$ & $[t_0, +\infty)$ & $[t_0, t_1)$ \\
		$fr$ & $\emptyset$ & $[t_0, +\infty)$ & $(t_2, +\infty)$ & $[t_0, +\infty)$ \\
		$lb$ & $[t_0, +\infty)$ & $(t_2, +\infty)$ & $[t_0, +\infty)$ & $\emptyset$ \\
		$lf$ & $[t_0, t_1)$ & $[t_0, +\infty)$ & $\emptyset$ & $[t_0, +\infty)$ \\
		\cmidrule(lr){2-5} 
	\end{tabular}
\end{table}

\begin{figure}
\centering \includegtikz[width=.9\columnwidth]{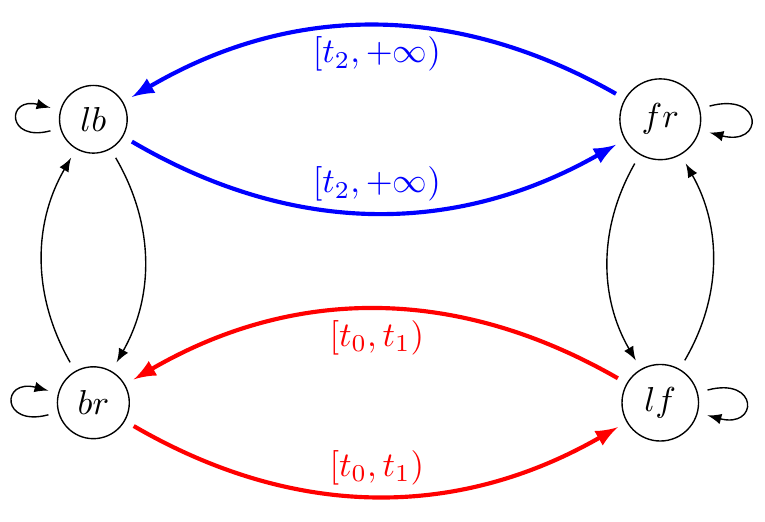}
\caption{Transition graph corresponding to \Cref{fig:partition-3d}, with validity set of each edge. Thinner edges shown in black have a validity set $[t_0, +\infty)$ (omitted for readability). \label{fig:continuous-tg}}
\end{figure}

Using \Cref{tab:adjacency-set}, we can build a directed graph (that we call \textit{transition graph}) representing all the possible transitions between cells of the partition as shown in \Cref{fig:continuous-tg}: each vertex of this graph corresponds to a partition cell, and we add the edge $A \rightarrow B$ if $\Adj(A,B) \neq \emptyset$. Additionally, we associate to each edge of the graph the corresponding validity set. A path in this graph is given as a succession of edges and associated transition times within the validity set of each edge, for instance $\left((br \rightarrow lb, t_1), (lb \rightarrow fr, t_2) \right)$ corresponding to the maneuver of waiting for the green vehicle ($2$) to pass before overtaking the blue one ($1$). Between these explicit transition times, the \ev is supposed to remain inside the last reached cell.

Using this graph-based representation also allows to compute a risk metric associated to a maneuver, called \textit{time margin}. This measure is defined as the time which remains to the \ev to perform a particular maneuver, before the most constrained transition becomes impossible. To illustrate this notion (which is formally defined in \Cref{sec:math}), we present example time margins for a selection of paths in \Cref{tab:time-margin}. 

\begin{table}
	\centering
	\caption{Time margins of example paths \label{tab:time-margin}}
	\begin{tabular}{l c c}
		\toprule
		Path & Most constr. trans. & Margin \\
		\midrule 
		$\left( (br \rightarrow br, t_0) \right)$ & $br \rightarrow br$ & $+\infty$ \\ 
		$\left((br \rightarrow lb, t_1), (lb \rightarrow fr, t_3) \right)$ & $lb \rightarrow fr$ & $+\infty$ \\ 
		$\left((br \rightarrow lf, t_0), (lf \rightarrow fr, t_1) \right)$ & $br \rightarrow lf$ & $t_1 - t_0$\\ 
		\bottomrule
	\end{tabular}
\end{table}

\begin{figure}
	\centering \subfloat{\includegtikz[width=.7\columnwidth]{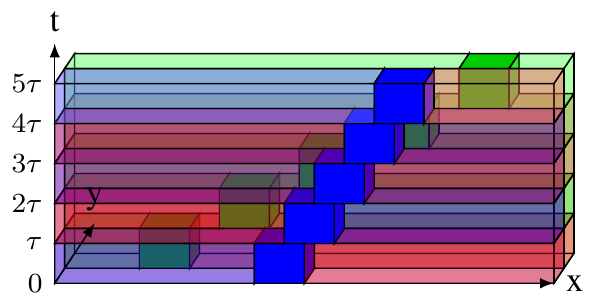}}
	\subfloat{\includegtikz[width=.25\columnwidth]{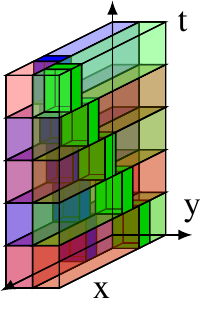}}
	\caption{Discrete partitioning of the free space-time of \Cref{fig:free-tspace}, with $t_0 = 0$. \label{fig:partition-3d-discrete}}
\end{figure}

\begin{figure}
	\centering \includegtikz[width=.85\columnwidth]{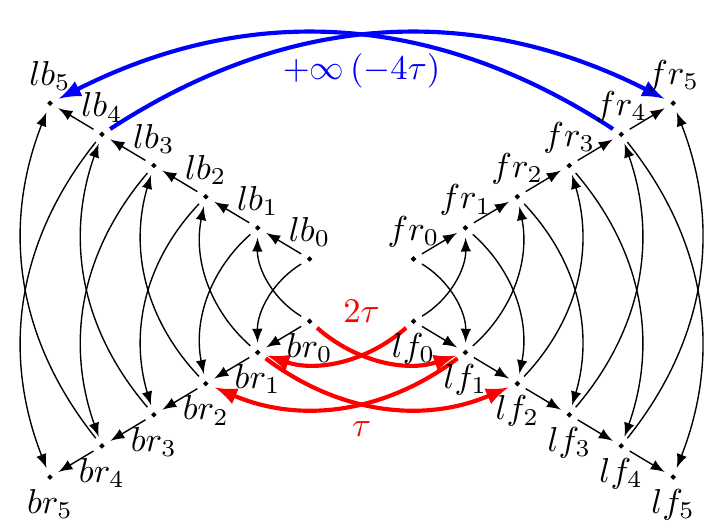}
	\caption{Discrete-time transition graph and time margins corresponding to the partition of \Cref{fig:partition-3d-discrete}. Vertex $A_k$ corresponds to the \ev being in set $A$ at time $t_0 + k\tau$. \label{fig:transition-graph}}
\end{figure}

Although this continuous approach is mathematically interesting, it is not necessarily suited for practical computer implementation, which is generally based on time sampling. For this reason, we also propose a discrete partitioning as shown in \Cref{fig:partition-3d-discrete}: for a discretization time step $\tau > 0$, we approximate the free space as a union of disjointed cylinders of the form $A \times [t_0 + k\tau, t_0 + (k+1)\tau)$ where $A$ is a subset in the partition at time $t_0 + k\tau$. Using this time-discretized partition, we can adapt the notion of adjacency to design a time-discretized transition graph, as shown in \Cref{fig:transition-graph}. In this graph, a path can be simply given as a list of successive vertices, thus allowing to use classic exploration algorithms. The time margin of any edge in the graph can also be easily computed (as shown in \Cref{sec:math}). Note that it is also possible to perform an event-based (instead of constant-time-based) partition, which has the advantage of exactly matching the time instants when the adjacency of two cells changes. However, since the partitioning will ultimately be used using the \ev's feasible dynamics -- which are not easily integrated into an event-based framework -- the constant-time discretization is preferred.

We believe that the proposed graph-based representation has two main advantages. First, the combinatorial part of the trajectory planning problem, consisting in choosing a feasible maneuver around the obstacles, is reduced to selecting a path in a transition graph. We will show in \Cref{sec:optimization} that, once such a path is given, computing a corresponding optimal trajectory becomes relatively simple for a large class of cost functions. Second, the graph approach makes it easy to take into account safety margins by avoiding exploration of time-constrained edges, which can be useful to handle uncertainty in trajectory estimation. Additional metrics can also be computed (see, \eg,~\cite{Constantin2014}) for spatial constraints, in order to account for control or positioning error.


\section{Mathematical results}\label{sec:math}
\subsection{Modeling}
We now proceed to theorize and generalize the intuitions exposed in the previous section. We consider an autonomous ego ground vehicle, driving on a road in presence of obstacles, which can either be fixed or mobile. The \ev is assumed to remain parallel to the local direction of the road, as it usually is the case in normal (non-crash) situations, so that its configuration is given by the position of its center of mass, denoted by $(x,y)$ in ground coordinates. 

\begin{figure}
	\centering \includegtikz[width=0.7\columnwidth]{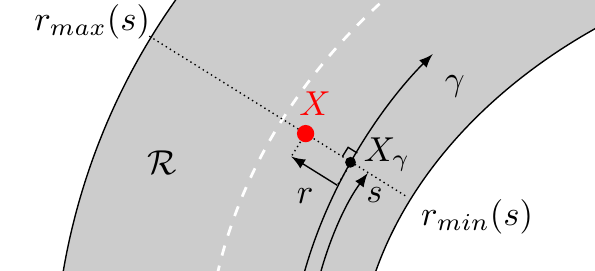}
	\caption{Frenet coordinates of a point on the road. \label{fig:frenet}}
\end{figure} 

We assume that the \ev has knowledge of the road geometry, for instance through cartography, as a $\mathcal C^2$ reference path $\gamma$ and bounds on the lateral deviation from $\gamma$, as shown in \Cref{fig:frenet}; we let $\RR \subset \mathbb R^2$ be such that the \ev is on the road if, and only if, $(x,y) \in \RR$. Finally, we assume that the road curvature and width are such that, for all $X = (x,y) \in \RR$, there exists a unique point $X_\gamma \in \gamma$ which is closest to $X$. 

According to \Cref{fig:frenet}, we define the Frenet coordinates of $X$ as $(s(X), r(X))$, where $s(X)$ is the curvilinear position of the corresponding point $X_\gamma$ along $\gamma$, and $r(X) = (X - X_\gamma) \cdot \bN$ with $(\bT,\bN)$ the Frenet frame of $\gamma$ at point $X_\gamma$. With these notations, we let $r_{min}$ and $r_{max}$ be such that $X  \in \RR$ if, and only if, $r_{min}(s(X)) \leq s(X) \leq r_{max}(s(X))$, and we let $\QQ = \left\{ (s,r) \in \mathbb R^2 \ : r_{min}(s) \leq r \leq r_{max}(s) \right\}$ denote the extent of the road in Frenet coordinates. In what follows, we only consider the Frenet coordinates of the \ev, and we drop the dependence of $s$ and $r$ in $X$. We assume that the \ev only moves forward along the road, in the direction of increasing $s$.

We denote by $\mathcal O$ the set of obstacles (considered as open sets) existing on the road around the \ev, and by $N = |\mathcal O|$ the number of obstacles. At a given time $t_0$, we consider a time horizon $T$ and we assume that an estimation of the trajectory of each obstacle $o \in \mathcal O$ is available over $[t_0, t_0+T]$. This estimation could come, for instance, be performed through machine learning techniques~\cite{altche2017}. We let $\chi = \QQ \times [t_0, t_0+T]$ the set of space-time points for which the vehicle is on the road, and we define the free portions of the space (respectively, of the space-time) as follows:

\begin{defn}[Free space]
	The (collision-)free space at time $t$ is the set $\QQ_f^t = \QQ \setminus \bigcup_{o \in \mathcal O} o$.
	The (collision-)free space-time over $[t_0,t_0+T]$ is the $\chi_f = \left\{ \QQ_f^t \times {t} \ : \ t \in [t_0, t_0+T]\right\}$.
\end{defn}
	
We call obstacle space-time $\chi_o$ the complement of $\chi_f$ in $\chi$. Note that the free space-time is similar to the notion of configuration space(-time), which is widely used in robotics~\cite{erdmann1986multiple}, and can be computed efficiently~\cite{kockara2007collision} provided that each obstacle's trajectory is known in advance.
In this article, we suppose perfect knowledge of these future trajectories over $[t_0, t_0+T]$; however, probabilistic trajectory estimates can also be taken into account, for instance by defining $\chi_f^p$ as the set of points of $\chi$ which are free with probability $p$.

To simplify the rest of the presentation, we consider that for all $t_1 \in [t_0, t_0+T]$, the intersection of the obstacle space-time $\chi_o$ is a union of (potentially rotated) rectangles; due to the roughly rectangular shape of classical vehicles, this assumption does not excessively sacrifice precision. Moreover, we assume that the road boundary functions $r_{min}$ and $r_{max}$ are piecewise-linear and continuous. In the following subections, we present our approach to partition $\chi_f$ into semantically meaningful subsets using a two-step algorithm: first, we partition planes corresponding to a fixed time $t_1 \in [t_0, t_0+T]$ in \Cref{sec:2dpartition}; second, we deduce a partition of $\chi_f$ in \Cref{sec:3dpartition,sec:3dpartitiondiscrete}.


\subsection{Semantic free-space partitioning}\label{sec:2dpartition}

\begin{figure}
	\centering \includegtikz[width=0.82\columnwidth]{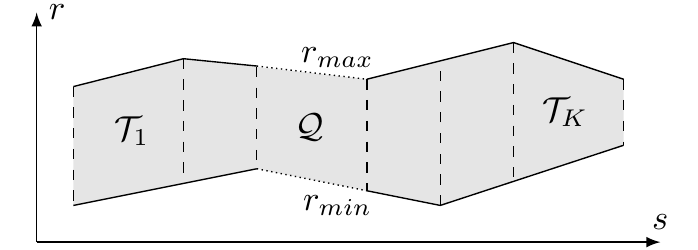}
	\caption{Decomposition of the road (in grey) in trapezes. \label{fig:trap-decomp}}
\end{figure} 

First, note that we can use trapeze decomposition to partition the road in convex regions using $r_{min}$ and $r_{max}$ as shown in \Cref{fig:trap-decomp}; since the road profile does not depend on time, this decomposition allows to fully partition $\chi$ by using cylinders with trapezoidal base. Each trapeze $\mathcal T_k$ (with $k \in \{1\twodots K\}$) can be defined by a set of linear constraints\footnote{\label{note:strict-ineq}To ensure the sets are disjoint, some of the inequalities should be strict. In practice, we use non-strict inequalities with a small tolerance $\varepsilon$.}, in the form $A_k X \leq b_k$ with $A_k$ a 4-by-2 matrix, $X = [s, r]^T$ and $b_k$ a vector of $\mathbb R^4$. This approach allows modeling varying roadway width and curvature, but may require an important number of trapezes to correctly handle sharp bends.

\begin{figure}
	\centering \includegtikz[width=0.8\columnwidth]{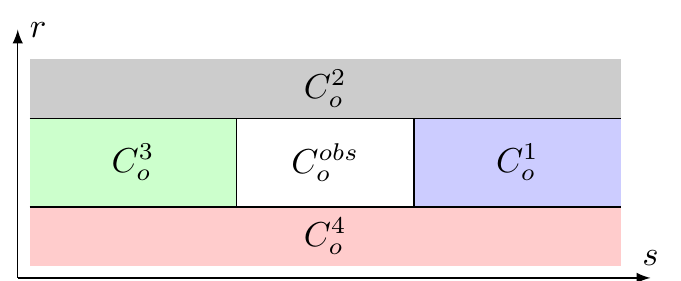}
	\caption{Partitioning of the 2D space around a single obstacle ($C_o^{obs}$) into four collision-free regions $C_o^i$. \label{fig:partition1v}}
\end{figure} 

For a single rectangular obstacle $o$ at time $t_1$, we define four regions $C_o^i \subset \mathbb R^2$ ($i \in \{1\twodots 4\}$) as illustrated in \Cref{fig:partition1v}; as in \Cref{sec:example}, these regions can be identified as positions where the \ev is located in front, to the left, behind or to the right of the obstacle. Similarly, we let $C_o^{obs}$ be the obstacle region corresponding to $o$. Since all obstacles are assumed rectangular, each of the $C_o^i$ regions is defined by a set of linear constraints\textsuperscript{\ref{note:strict-ineq}} in the form $A_o^i X \leq b_o^i$, with $A_o^i$ a two-column matrix, $X = [s, r]^T$ and $b_o^i$ a vector having the same number of lines as $A_o^i$. The partition of the free space $\QQ_f^{t_1}$ can be built recursively according to \Cref{alg:partitioning}; \Cref{thm:partition} ensures the validity of this algorithm; we let $\PP^{t_1}$ be the partition of $\QQ_f^{t_1}$ obtained by \Cref{alg:partitioning}.

\begin{theorem}[Partition]\label{thm:partition}
	$\PP^{t_1}$ is a partition of $\QQ_f^{t_1}$.
\end{theorem}
\begin{proof}
	We will prove that, for all $0 \leq n \leq N$, $\PP_n$ is a partition of $\QQ_n = \QQ \setminus \bigcup_{i = 1}^n C_{o_i}^{obs}$. First, this property is verified for $\PP_0$ which is a partition of $\QQ$. Second, the loop preserves the following invariants for all $n \geq 1$ and $e \in \PP_n$: 
	\begin{itemize}
		\item $e \neq \emptyset$ and $\exists e' \in \PP_{n-1}$ such that $e \subset e'$;
		\item for all $1 \leq i \leq n$, $\exists j \in \{1 \twodots 4\}$ such that $e \subset C_{o_i}^j$.
	\end{itemize}
	Thus, $\PP_n = \left\{ e \cap C_{o_n}^j \middle| e \in \PP_{n-1}, j \in \{1 \twodots 4\}, e \cap C_{o_n}^j \neq \emptyset\right\}$. Since the sets $\left( \mathcal C_{o_n}^j \right)_{j = 1\twodots 4}$ define a partition of $\mathbb R^2 \setminus C_{o_n}^{obs}$ and since all $e \in \PP_0$ is a subset of $\QQ$, we deduce by induction that all elements of $\PP_n$ are nonempty subsets of $\QQ_n$. 
	
	Reciprocally, for all any $q \in \QQ_n = \QQ_{n-1} \setminus C_{o_n}^{obs}$ there exists $j \in \{1 \twodots 4\}$ such that $q \in \QQ_{n-1} \cap C_{o_n}^j$. Since $\PP_0$ is a partition of $\QQ$, inductive reasoning yields $ \QQ_n \subset \bigcup_{e \in \PP_n} e$.
\end{proof}

\begin{algorithm}
\caption{Partitioning of $\QQ_f^{t_1}$\label{alg:partitioning}}
\begin{algorithmic}
	\State $\PP_0 \gets \{\mathcal T_k\}_{k=1\twodots K}$ \Comment{Initialize $\PP_0$ as a partition of $\QQ$} 
	\State $N \gets |\mathcal O|$
	\For {$n = 1\twodots N$} \Comment{Loop over all obstacles $o_n$}
	\State $\PP_n \gets \{\}$
	\ForAll{$C \in \PP_{n-1}$} \Comment{Loop over cells $ C$ in $\PP_{n-1}$}
	\For{$j = 1\twodots4$}
	\If{$C_{o_n}^j \cap C \neq \emptyset$}  \Comment{Partition $ C \setminus C_{o_n}^{obs}$}
	\State $\PP_n \gets \PP_n \cup \{ C_{o_n}^j \cap C \}$
	\EndIf
	\EndFor
	\EndFor
	\EndFor
	\State $\PP^{t_1} \gets \PP_N$
\end{algorithmic}
\end{algorithm}

From the previous proof, we deduce that our partitioning of $\QQ_f^{t_1}$ bijectively corresponds to relative positions from all $N$ obstacles in the free space at time $t_1$. Thus, each element in the partition can be uniquely defined by a signature, as stated in \Cref{def:cor-signature} and \Cref{def:signature}:

\begin{corollary}[Semantization]\label{def:cor-signature}
	For all $e \in \PP^{t_1}$, there exists a unique tuple $\sigma_{t_1}(e) = (k, j_1, \threedots, j_N) \in \{1 \twodots K\} \times \{1 \twodots 4\}^N$ such that $e = \mathcal T_k \cap \bigcap_{n = 1\twodots N} C_{o_n}^{j_n}$. Using $\Sigma = \{1 \twodots K\} \times \{1 \twodots 4\}^N$, $\sigma_{t_1}$ is a bijection from $\Sigma$ to $\PP^{t_1} \cup \emptyset$.
\end{corollary}

\begin{defn}[Signature]\label{def:signature} We call $\sigma_{t_1}(e) \in \Sigma$ from \Cref{def:cor-signature} the \textit{signature} of subset $e$.
\end{defn}

Moreover, there is a finite number of elements in the partition which is bounded by $K 4^N$ for $K$ trapezes and $N$ obstacles. Additionally, all elements $e \in \PP^{t_1}$ also are convex polygons (or \textit{cells}), which can be fully described using a single (matrix, vector) pair that can easily be stored in computer memory. \Cref{fig:example3v,fig:partition3v} illustrate our partitioning in a more complex scenario\footnote{Obstacle regions $1$, $2$, $3$ in \Cref{fig:partition3v} are computed for a point-mass \ev, in order to match the vehicle shapes shown in \Cref{fig:example3v}.} with $3$ vehicles; note that, for clarity purposes, we respectively used $f, l, b, r$ instead of $1, 2, 3, 4$ as defined in \Cref{def:signature}. Also remark that, although \Cref{fig:example3v} is shown in world coordinates $(x,y)$, \Cref{fig:partition3v} uses Frenet coordinates $(s,r)$. In order to encode the relation between elements of the partition, we introduce the notion of adjacency as follows:
\begin{defn}[Adjacency]
	\label{def:adjacency}
	For $e_1, e_2 \in \mathcal P^{t_1}$, we say that $e_1$ and $e_2$ are adjacent if, and only if the intersection of their closures is not empty, \ie $\overline{e_1} \cap \overline{e_2} \neq \emptyset$. For $\sigma, \sigma' \in \Sigma$, we let $\adj_{t_1}(\sigma, \sigma') = 1$ if $\sigma_{t_1}^{-1}(\sigma)$ and $\sigma_{t_1}^{-1}(\sigma')$ are adjacent, and $0$ otherwise.
\end{defn}

Note that this definition considers cells whose closures intersect at single point as adjacent; this situation could happen, \eg, in the case shown in \Cref{fig:partition3b} between (br) and (lf). From a theoretical standpoint, the hypothesis that obstacles are open sets allows to overcome this issue since, in this case, the vehicle can actually perform the maneuver; in practice, the use of safety margins around the obstacles, and constraints on time margin (see \Cref{sec:3dpartitiondiscrete}) prevent this question from becoming an issue. The main reason for choosing this slightly looser criterion compared, \eg, to that of~\cite{Park2015} -- requiring the intersection to be a non-singleton segment -- is that it can be verified in polynomial time using the matrix inequality representation and linear programming.  


\begin{figure}[h]
	\centering \includegtikz[width=0.8\columnwidth]{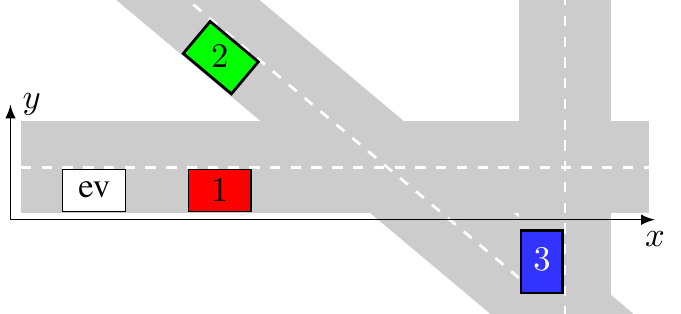}
	\caption{A more complicated example with 3 obstacles. \label{fig:example3v}}
\end{figure} 

\begin{figure}[h]
	\centering \includegtikz[width=0.9\columnwidth]{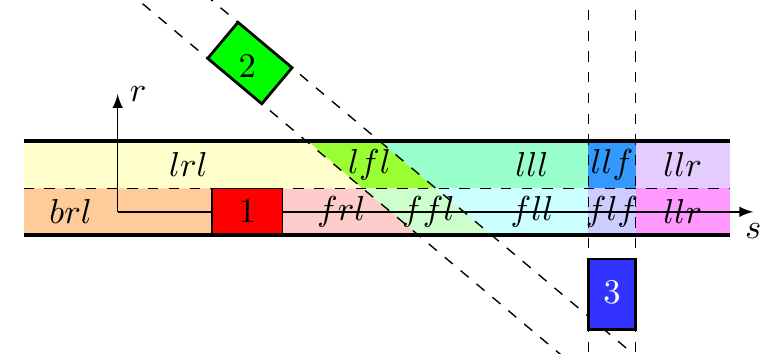}
	\caption{Partitioning in the example of \Cref{fig:example3v}, with subsets signature; for instance, $brl$ means that the \ev is behind the first vehicle, to the right of the second and left of the third. The thick black lines correspond to the unique trapeze encoding road boundaries; its index is omitted for clarity. \label{fig:partition3v}}
\end{figure} 

\subsection{Continuous free space-time partitioning}\label{sec:3dpartition}
We now proceed to generalize these results to the free space-time $\chi_f$; to define a partition of this space we use the union (over time) of cells sharing the same signature:
\begin{defn}[Space-time cell]For $\sigma \in \Sigma$, we let $E_\sigma = \bigcup_{t_1 \in [t_0, t_0+T]} \sigma_{t_1}^{-1}(\sigma) \times \{t_1\}$ be the space-time cell corresponding to $\sigma$, \ie the set of all points in the free space-time sharing this signature.
\end{defn}

Since the set of non-empty elements of $\sigma_t^{-1}(\Sigma)$ defines a partition of $\QQ_f^t$, the set $\PP = \left\{ E_\sigma\, \middle|\, \sigma \in \Sigma, E_\sigma \neq \emptyset \right\}$ defines a partition of $\chi_f$ (see \Cref{fig:partition-3d}). The notion of adjacency can then be generalized as follows:
\begin{defn}[Validity set]
	For $\sigma$, $\sigma' \in \Sigma$, we define the validity set $\Adj(\sigma, \sigma') = \left\{ t \in [t_0, t_0+T]\, \middle|\, \adj_t(\sigma, \sigma') = 1 \right\}$.
\end{defn} 
\vspace*{-0.5cm}

Using this notion, we can now define a transition graph (see \Cref{fig:continuous-tg}) as follows:
\begin{defn}[Continuous transition graph]
	The continuous transition graph is the directed graph $\mathcal G^c = (\mathcal V^c, \mathcal E^c, \Adj)$ with vertex set $\mathcal V^c = \left\{\sigma \in \Sigma\, \middle|\, E_\sigma \neq \emptyset \right\}$, edges set $\mathcal E^c = \left\{ (\sigma_1, \sigma_2) \in \Sigma^2\, \middle|\, \Adj(E_{\sigma_1}, E_{\sigma_2}) \neq \emptyset \right\}$ and associated validity set $\Adj$.
\end{defn}

The motivation for introducing this graph is that any collision-free maneuver corresponds to a unique \textit{path} in $\mathcal G^c$. To account for the temporal aspect of this graph, such a path is defined as follows:
\begin{defn}[Path in $\mathcal G^c$]
	A path in $\mathcal G^c$ is given by a list of vertices $(\sigma_1, \threedots, \sigma_{m+1}) \in \mathcal V^c$ so that for all $i \leq m$, $(\sigma_i, \sigma_{i+1}) \in \mathcal E^c$, and a list of strictly increasing transition times $(t_1, \threedots, t_m)$ such that for all $i \in \{1 \twodots m\}$, $t_i \in \Adj(\sigma_i, \sigma_{i+1})$ and $[t_i, t_{i+1}) \subset \Adj(\sigma_i, \sigma_i)$.
\end{defn}

In other words, a path in $\mathcal G^c$ is a sequence of cells and time instants corresponding to the transition time between two successive cells; between two consecutive transitions, the \ev can remain in the cell it occupies last. For a given path, we can now define the corresponding time margin as the time left for the vehicle to perform the most constrained transition before it becomes impossible:
\begin{defn}[Time margin along a path]\label{def:validity-continuous}
	Consider a path in $\mathcal G^c$ given as $\pi^c = \big( (\sigma_1, \twodots, \sigma_{m+1}), (t_1, \twodots, t_m) \big)$. The time margin along $\pi^c$ is $V(\pi^c) = \min_{i = 1\threedots m} \left( \sup \Big\{ t - t_i\ \big|\ [t_i, t) \subset \Adj(\sigma_i, \sigma_{i+1}) \Big\} \right)$.
\end{defn}

\subsection{Discrete-time partitioning}\label{sec:3dpartitiondiscrete}
Due to the potentially complex trajectories followed by the obstacles, there is no guarantee regarding the topology of the subsets $E_\sigma$, which can for instance have multiple connected components; similarly, $\Adj(\sigma_1, \sigma_2)$ is in general a union of disjointed intervals. To make practical applications easier, we also propose a temporal discretization of the free space-time with a time step duration $\tau$ (with $T = P\tau$). Note that the value of $\tau$ depends on a trade-off between acceptable computation time, length of the planning horizon and required precision in vehicle dynamics; \Cref{sec:simulation} provides some performance reports on the influence of this parameter. We approximate $E_\sigma$ as a union of box-shaped cells:

\begin{defn}[Discrete space-time cell]
	For $\sigma \in \Sigma$ and $p \in \{0 \threedots P\}$, we let $E_\sigma^p = \sigma_{\theta_p}^{-1}(\sigma) \times [\theta_p, \theta_{p+1})$ be the discrete space-time cell corresponding to $\sigma$ at step $p$, with $\theta_p = t_0 + p\tau$.
\end{defn}

In the rest of this article, we assume\footnote{This assumption requires taking slight safety margins, the size of which being proportional to the distance traveled by the obstacles during the duration of a discretization time step $\tau$.} that $E_\sigma^p \subset \chi_f$ for all $\sigma \in \Sigma$ and $p \in \{0 \twodots P\}$. Since $\sigma_{\theta_p}^{-1}(\sigma)$ is a convex (or empty) set, $E_\sigma^p$ is either empty or convex, and fully defined by a set of linear inequalities in the form $A[X, t]^T \leq b$ (the comments of \cref{note:strict-ineq} also apply here). Finally, we define a partition of the free space-time $\chi_f$ in convex box-shaped cells as $\PP^{\tau} = \left\{ E_\sigma^{p} \middle| p \in \{0\twodots P\}, \sigma \in \Sigma, E_\sigma^{p} \neq \emptyset \right\}$ (see \Cref{fig:partition-3d-discrete}), and we associate a discrete transition graph:


\begin{defn}[Discrete transition graph]
	\label{def:transition-graph}
	The discrete transition graph is the directed graph $\mathcal G^d = (\mathcal V^d, \mathcal E^d)$ with vertex set $\mathcal V^d = \PP^\tau$ and edges set $\mathcal E^d = \left\{ (E_{\sigma_1}^p, E_{\sigma_2}^{p+1}) \middle| E_{\sigma_1}^p, E_{\sigma_2}^{p+1} \in \mathcal V^d,\, \adj_{\theta_p}(\sigma_1, \sigma_2) = 1 \right\}$.
\end{defn}

Therefore, each vertex of $\mathcal G$ corresponds to a certain cell of the partition $\PP^\tau$ at a given time $\theta_p$, and the edge $v_1 \rightarrow v_2$ exists if $v_1$ and $v_2$ represent two adjacent cells (possibly twice the same) at two consecutive time steps. Paths in $\mathcal G^d$ comply with the usual definitions of graph theory and can be given as a set of vertices.

Finally, we define the time margin for paths in $\mathcal G^d$ as:
\begin{defn}[Time margin in the discrete graph]
	For a path $\pi^d = \big( E_{\sigma_0}^{0}, \threedots, E_{\sigma_{m+1}}^{m+1} \big)$ in $\mathcal G^d$, the time margin is $v(\pi^d) = \displaystyle{\min_{i = 0\threedots {m}}} \max_{p = i\dots m} \tilde v(i,p)$ with $\tilde v(i,p) = \big\{ \tau(p - i+1)\, \big|\, \forall q \in \{ i \twodots p \}, \adj_{\theta_q}(\sigma_i, \sigma_{i+1}) = 1 \big\}$.
\end{defn}
Note that $\tilde v(i,p)$ is the discrete-time equivalent of the set $\left\{ t - t_i\ \big|\ [t_i, t) \subset \Adj(\sigma_i, \sigma_{i+1}) \right\}$ from \Cref{def:validity-continuous}, and this discrete time margin is analogous to that of \Cref{def:validity-continuous}.

\section{Application to planning and control}\label{sec:optimization}
Before presenting the applications of our approach to planning and control, let us formalize the link between paths in the transition graph and trajectories for the \ev:
\begin{defn}[Interpretation of transition graph paths]
	Let $\pi_0^c = \big( (\sigma_1, \threedots, \sigma_{m+1}), (t_1, \threedots, t_m) \big)$ be a path in $\mathcal G^c$ and $x(t)$ be a collision-free trajectory for the \ev. We say that $\pi_0$ corresponds to $x$ if, for all $i \in \{1 \twodots m\}$, $x \big([t_{i-1}, t_i) \big) \subset E_{\sigma_i}$ and $x \big([t_m, T) \big) \subset E_{\sigma_{m+1}}$.
\end{defn}

From this definition, we deduce that for a given collision-free trajectory $x(t)$ there exists a unique corresponding path in $\mathcal G^c$ denoted by $\pi(x)$. Reciprocally, for a given path $\pi_0$ in $\mathcal G^c$, there exists a set of corresponding trajectories denoted by $\pi^{-1}(\pi_0)$. We obtain the following theorem:

\begin{theorem}[Motion-planning equivalence]
	Let $J(x)$ be a cost function for a given trajectory $x(t)$, $X$ the set of collision-free trajectories, and $\Pi$ the set of paths in $\mathcal G^c$. Then:
	\begin{equation}
		\min_{x \in X}  J(x) = \min_{\pi_0 \in \Pi}\left( \min_{x \in \pi^{-1}(\pi_0)} J(x) \right)
	\end{equation}
\end{theorem}\vspace*{-0.5cm}
\begin{proof}
	From the previous definition, for any collision-free trajectory $x \in X$ there exists $\pi_0 \in \Pi$ such that $\pi(x) = \pi_0$. Therefore, $\min_{x \in X}  J(x) \leq \min_{\pi_0 \in \Pi}\left( \min_{x \in \pi^{-1}(\pi_0)} J(x) \right)$. Reciprocally, for all $\pi_0 \in \Pi$, any $x \in \pi^{-1}(\pi_0)$ is guaranteed to be collision-free, leading to the reciprocal inequality.
\end{proof}

In other words, it is equivalent to find an optimal trajectory for the \ev, and to find an optimal path in the transition graph $\mathcal G^c$ and then the optimal trajectory corresponding to this path. These results can be extended to paths in the discrete graph $\mathcal G^d$; due to space limitations, details are not presented here and we only provide the following definition:
\begin{defn}[Interpretation of discrete graph paths]
	Let $\tau > 0$ be a discretization time step, $\pi_0^d = (E_{\sigma_0}^0, \threedots, E_{\sigma_{m+1}}^{m+1})$ be a path in $\mathcal G^d$ and $x(t)$ be a collision-free trajectory for the \ev. We say that $\pi_0^d$ corresponds to $x$ if, for all $p \in \{0 \twodots m\}$, $x(\theta_p) \in E_{\sigma_p}^p$ and $x([\theta_p, \theta_{p+1})) \subset E_{\sigma_p}^p \cup E_{\sigma_{p+1}}^p$.
\end{defn}

An interesting feature of this decomposition of the trajectory planning problem is that we effectively separate the discrete choice of a maneuver variant, and the search for an optimal control corresponding to this maneuver as was obtained in~\cite{Park2015} for path-planning. This second problem can be solved efficiently under certain assumptions on the vehicle dynamics. We consider an AGV with linear discrete dynamics $x_{p+1} = Ax_p + Bu_p$ for a state $x_p$ and a control $u_p \in U$ (with $U$ a convex polyhedron) with a discretization time step $\tau$, where $A$ and $B$ have constant coefficient. For a positive semi-definite matrix $Q$, a line vector $L$ and defining $X_p = [x_p^T, u_p^T]^T$, we consider a generic quadratic cost function $J(x,u) = \frac 1 2 X_p^T Q X_p + L^T X_p$. In this case, an optimal solution can be computed in polynomial time:

\begin{theorem}[Polynomial time computability]\label{thm:polynomial-time}
	Let $\pi_0$ be a path in $\mathcal G^d$ and $x_0$ an initial AGV state. The optimal trajectory (and associated control sequence) $(X_p)$ starting from $x_0$ realizing $\displaystyle \min_{(x_p) \in \pi^{-1}(\pi_0), u_p \in U} J(x,u)$ can be computed in polynomial time in the number of obstacles and time steps.
\end{theorem}
\begin{proof}
	We will show that this problem is an instance of convex quadratic programming (QP), which has a complexity $\mathcal O(n^3)$ where $n$ is the number of constraints~\cite{Vavasis2009}. First, the cost function $J$ is quadratic and convex. Second, vehicle dynamics and control bounds can be encoded as linear constraints. Moreover, the condition $(x_p) \in \pi^{-1}(\pi_0)$ corresponds to a set of $\mathcal O(PN)$ linear constraints, leading to a QP problem with complexity  $\mathcal O \left((PN)^3 \right)$ for $N$ obstacles over $P$ time steps, thus proving the announced result.
\end{proof}

\section{Simulation results\label{sec:simulation}}
To showcase the advantages of our approach and the axes for improvement, we present preliminary simulation results in the scenario of \Cref{fig:example-sit}. For illustration purposes, we consider a very simple second-order dynamics for the \ev as:
$	X_{p+1} = \left( \begin{smallmatrix}0 & I_2 \\ 0 & 0 \end{smallmatrix} \right) X_p + \left( \begin{smallmatrix}0 \\ I_2\end{smallmatrix} \right) u_p $
with $X_p = [s, r, \dot s, \dot r]_p^T$, $u_p = [a_{lon}, a_{lat}]_p^T$ and $I_2$ the $\mathbb R^2$ identity. To account for the nonholonomic constraints, we bound the lateral velocity as $|\dot r_p| \leq \alpha \dot s_p$ with $\alpha > 0$ a parameter, and we also require $a_{lon}$ and $a_{lat}$ to be bounded. The objective function  is chosen as $J = \sum_p \left(\dot s_p - s_p^{ref} \right)^2 + \dot r_p^2 + r_p^2$, and we use a planning horizon of \SI{10}{\second} with a time step $\tau = \SI{1}{\second}$, and a minimum time margin of \SI{1}{\second}. 

One difficulty in the proposed approach is the need to ensure that the selected optimal path in the transition graph is dynamically feasible. In this early implementation, we use a naive approach consisting in computing the optimal trajectory corresponding to each explored branch; as stated in \Cref{thm:polynomial-time}, this computation can be performed relatively fast. Additionally, we use a branch-and-bound technique to prune dynamically infeasible or poor quality branches. The trajectories from both algorithms are shown in \Cref{fig:traj}; \Cref{tab:comp-time} reports the computation time for all steps of our algorithm using a standard desktop computer. For comparison purposes, we added an implementation of the same problem using a pure mixed-integer quadratic programming (MIQP) method from~\cite{Qian2016}. We use an implementation of Seidel's linear programming algorithm~\cite{seidel1991small} to perform the partitioning, and Gurobi~\cite{gurobi} in version 7.5 to solve quadratic programming (trajectory optimization) and MIQP (from~\cite{Qian2016}) problems.

Although the MIQP is faster overall, it is not able to take notions such as time margin into account; this results in the \ev trying to overtake the blue vehicle (1) before the green one (2), which is a higher-risk maneuver that is excluded by our algorithm as shown in \Cref{fig:traj}; when removing the time margin constraints, both solvers converge to the same optimum. A likely explanation for the better performance of the pure MIQP method is that it can directly use vehicle dynamics to guide the exploration; future work will focus on designing an hybrid algorithm between these methods, to benefit from the advantages of both approaches.

\begin{figure}\centering
	{\includegtikz[width=0.95\columnwidth]{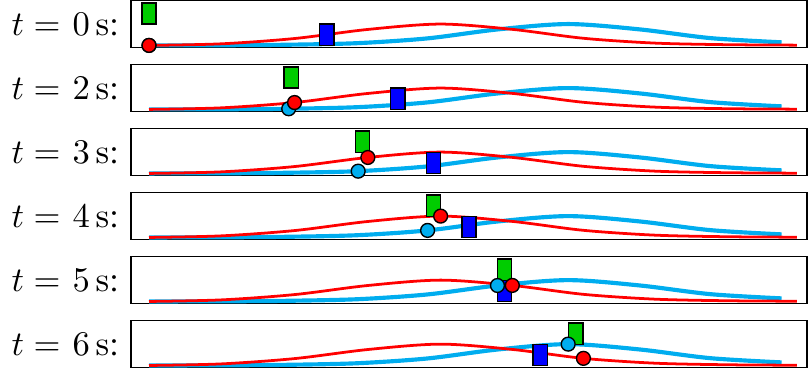}}
	\caption{Trajectories computed by both algorithms (MIQP in red, our graph-based approach in thicker cyan) for a time margin of \SI{1}{\second}. The rectangles correspond to the vehicles of \Cref{fig:example-sit}; the circles to the \ev position.\label{fig:traj}}
\end{figure} 

\begin{table}\centering
	\caption{Summary of computation time (average over 100 iterations)\label{tab:comp-time}}
	\begin{tabular}{l l l l}
		\toprule
		Algorithm & Comp. time & Objective val. \\
		\midrule
		Partitioning &\SI{2.5}{\milli\second} & - \\
		Graph exploration (\SI{1}{\second} margin) &  \SI{27.6}{\milli\second} & 13.89 \\
		Optimal path computation &  $< \SI{1}{\milli\second}$ & - \\
		\midrule
		MIQP~\cite{Qian2016} & \SI{14.0}{\milli\second} & 13.58 \\ 
		Our graph exploration (no margin) &  \SI{28.4}{\milli\second} & 13.58 \\
		\bottomrule
	\end{tabular}
\end{table}

\section{Conclusion\label{sec:conclusion}}
This article generalized the divide-and-conquer approach used in~\cite{Park2015} in two dimensions, to the 3D space-time for on-road navigation of autonomous ground vehicles in presence of moving obstacles. We described a systematic method to partition the collision-free space-time in the presence of fixed or moving obstacles, and we provided a graph representation of all possible collision-free trajectories. This approach allows to treat the combinatorial problem of optimal trajectory planning in two steps: first, a path-finding problem in a graph, and then a simple optimization that can be performed in polynomial time in the number of obstacles for any quadratic cost function. Moreover, we introduced a notion of time margin and showed that our graph-based approach can easily take into account margin of error in the execution for a particular maneuver. Coupled with additional similar metrics, we believe that our approach can have useful applications for planning under prediction and control uncertainty, notably in the frame of stochastic decision-making. Future work will also focus on improving our graph exploration algorithm to allow faster computation. 

\bibliographystyle{IEEEtranurl}
\balance
\bibliography{behaviorplanning}

\end{document}